\documentclass[aps,pra,groupedaddress,twocolumn,nofootinbib,superscriptaddress]{revtex4}

\usepackage{amsmath}
\usepackage{amssymb}
\usepackage{amsthm}
\usepackage{bm}
\usepackage{graphicx}
\usepackage{color}
\usepackage[hyperindex,breaklinks]{hyperref}

\usepackage{xcolor}

\newcommand{\RR}{\mathbb{R}}

\newcommand{\BC}{\mathcal{B}}

\newcommand{\HC}{\mathcal{H}}
\newcommand{\IC}{\mathcal{I}}

\newcommand{\RC}{\mathcal{R}}
\newcommand{\SC}{\mathcal{S}}

\newcommand{\ket}[1]{|#1\rangle}                  
\newcommand{\ip}[2]{\langle #1|#2\rangle}         
\newcommand{\matl}[3]{\langle #1|#2|#3\rangle}    
\newcommand{\Tr}{{\rm Tr}}                        
\newcommand{\vect}[1]{\boldsymbol{#1}}         

\newcommand{\ii}{\mathrm{i}}					  
\newcommand{\ee}{\mathrm{e}}  

\def\dya#1{|#1\rangle \langle#1|}

\def\tr{{\rm Tr}}

\long\def\ca#1\cb{} 

\newtheorem{theorem}{Theorem}
\newtheorem{lemma}{Lemma}

\newtheorem{example}{Example}

\newtheorem*{definition*}{Definition}

\begin{document}

\title{Universal Uncertainty Relations}

\author{Shmuel Friedland}
\email{friedlan@uic.edu}
\affiliation{Department of Mathematics, Statistics and Computer Science, University
of Illinois at Chicago, 851 S. Morgan Street, Chicago, IL 60607-7045, U.S.A.}

\author{Vlad Gheorghiu}
\email{vgheorgh@gmail.com}
\affiliation{Institute for Quantum Science and Technology and
Department of Mathematics and Statistics,
University of Calgary, 2500 University Drive NW,
Calgary, AB, T2N 1N4, Canada}
\affiliation{Institute for Quantum Computing, University of Waterloo, Waterloo, ON, N2L 3G1, Canada}

\author{Gilad Gour}
\email{gour@ucalgary.ca}
\affiliation{Institute for Quantum Science and Technology and
Department of Mathematics and Statistics,
University of Calgary, 2500 University Drive NW,
Calgary, AB, T2N 1N4, Canada}

\date{Version of \today}

\begin{abstract}
Uncertainty relations are a distinctive characteristic of quantum theory that impose intrinsic limitations on the precision with which physical properties can be simultaneously determined. The modern work on uncertainty relations employs \emph{entropic measures} to quantify the lack of knowledge associated with measuring non-commuting observables. However, there is no fundamental reason for using entropies as quantifiers; any functional relation that characterizes the uncertainty of the measurement outcomes defines an uncertainty relation. Starting from a very reasonable assumption of invariance under mere relabelling of the measurement outcomes, we show that Schur-concave functions are the most general uncertainty quantifiers. We then discover a fine-grained uncertainty relation that is given in terms of the majorization order between two probability vectors, 
\textcolor{black}{significantly extending a majorization-based uncertainty relation first introduced in  [M. H. Partovi, Phys. Rev. A \textbf{84}, 052117 (2011)].}
Such a vector-type uncertainty relation generates an infinite family of distinct scalar uncertainty relations via the application of arbitrary uncertainty quantifiers. Our relation is therefore universal and captures the essence of uncertainty in quantum theory. 

\end{abstract}

\maketitle
Uncertainty relations lie at the core of quantum mechanics and are a direct manifestation of the non-commutative structure of the theory. In contrast to classical physics, where in principle any observable can be measured with arbitrary precision, quantum mechanics introduces severe restrictions on the allowed measurement results of two or more non-commuting observables. Uncertainty relations are not a manifestation of the experimentalists' (in)ability of performing precise measurements, but are inherently determined by the incompatibility of the measured observables.

The first formulation of the uncertainty principle was provided by Heisenberg \cite{Heisenberg1927}, who noted that more knowledge about the position of a \emph{single} quantum particle implies less certainty about its momentum and vice-versa. He expressed the principle in terms of standard deviations
of the momentum and position operators
\begin{equation}\label{eqn1}
\Delta X\cdot\Delta P \geqslant \frac{\hbar}{2}.
\end{equation}
Robertson \cite{PhysRev.34.163} generalized Heisenberg's uncertainty principle to any two arbitrary observables $A$ and $B$ as
\begin{equation}\label{eqn2}
\Delta A\cdot\Delta B\geqslant \frac{1}{2} \vert\matl{\psi}{[A,B]}{\psi}\vert.
\end{equation}

A major drawback of Robertson's uncertainty principle is that it depends on the state $\ket{\psi}$ of the system. In particular, when $\ket{\psi}$ belongs to the null-space of the commutator $[A,B]$,  the right upper bound becomes trivially zero. Deutsch \cite{PhysRevLett.50.631} addressed this problem by providing an \emph{entropic} uncertainty relation (EUR) in terms of the Shannon entropies of any two non-degenerate observables, later improved by Maassen and Uffink \cite{PhysRevLett.60.1103} to 
\begin{equation}\label{eqn3}
H(A) + H(B) \geqslant -2\log c(A,B).
\end{equation}
Here $H(A)$ is the Shannon entropy \cite{CoverThomas:ElementsOfInformationTheory}
of the probability distribution induced by measuring the state $\ket{\psi}$ of the system in the eigenbasis $\{\ket{a_j}\}$ of the oservable $A$ (and similarly for $B$). The bound on the right hand side $c(A,B):=\max_{m,n}|\ip{a_m}{b_n}|$ represents the maximum overlap between the bases elements, and is independent of the state $\ket{\psi}$.

Recently the study of uncertainty relations intensified \cite{PhysRevLett.106.110506, PhysRevLett.108.210405}  (see also \cite{1367-2630-12-2-025009, birula-review} for recent surveys), and as a result various important applications have been discovered, ranging from security proofs for quantum cryptography \cite{Damgaard07atight,6157089,PhysRevLett.100.220502}, information locking \cite{PhysRevLett.100.220502}, non-locality \cite{Oppenheim19112010}, and the separability problem \cite{PhysRevLett.92.117903}.  There were also recent attempts to generalize uncertainty relations to more than two observables. For this case relatively little is known \cite{0305-4470-25-7-014,Sanchez1993233, PhysRevA.75.022319,PhysRevA.79.022104,wehner:062105}, as the authors investigated only particular instances of the problem such as mutually unbiased bases.

In most of the recent work on uncertainty relations, entropy functions like the Shannon and Renyi entropies are used to quantify uncertainty. However, in the context of the uncertainty principle, these entropies are by no reason the most adequate to use. Indeed, as we show here, other functions can be more suitable in providing a quantitative description for the uncertainty principle. 
\textcolor{black}
{
Our approach is based on using majorization \cite{MarshallOlkin:Majorization} to quantify uncertainty. The idea of using majorization to study uncertainty relations was first introduced in \cite{PhysRevA.84.052117}, and here we build on these ideas and provide explicit closed formulas.
}

Uncertainty is related to the ``spread" of a probability distribution, or, equivalently, to the ability of learning that probability distribution. Intuitively a less spread distribution is more certain than a more widely spread. For example, in a $d$-dimensional sample space, the probability distribution $\vect{p}=(1,0,\ldots,0)$ is the most certain, whereas the distribution $\vect{q}=(1/d,1/d,\ldots, 1/d)$ is the most uncertain. What are then the minimum requirements that a good measure of uncertainty has to satisfy?

In his seminal paper  \cite{PhysRevLett.50.631} on EURs, Deutsch pointed out that the standard deviation $\Delta$ can be increased by mere relabelling of the random variables associated with the measurements. He therefore concluded that the relation in~\eqref{eqn1} can not be used as a quantitative description of the uncertainty principle. 

Following Deutsch observation,
we assume here that the uncertainty about a random variable can not decrease under a relabelling of its alphabet, i.e. the uncertainty associated with a probability vector $\vect{p}$ can not be larger than the uncertainty associated with a relabelled version of it, $\pi\vect{p}$, where $\pi$ is some permutation matrix. In fact, both uncertainties are the same as permutations acting on a probability space are reversible. Next, we make the reasonable assumption that uncertainty can not decrease by forgetting information (discarding), see Fig.~\ref{fgr1}. 
We call this very reasonable presumption \emph{monotonicity under random relabelling} (MURR). This will be our \emph{only} requirement for a measure of uncertainty.
We therefore conclude that any reasonable  measure of uncertainty is a function only of the probability vector, is invariant under permutations of its elements, and must be non decreasing under a random relabelling of its argument. 
\begin{figure}
\includegraphics[scale=0.35]{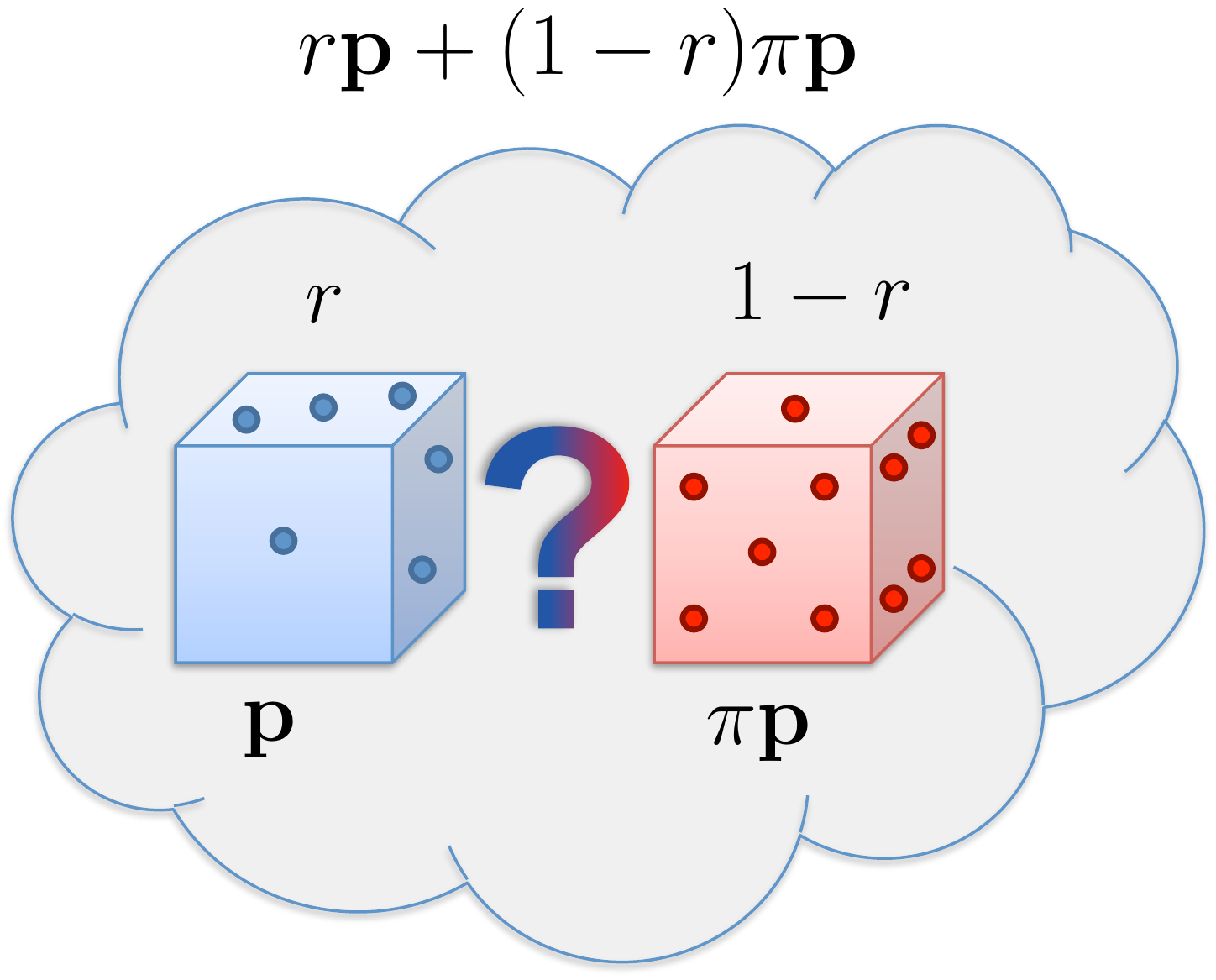}
\caption{With probability $r$ Alice samples from a random variable (blue dice), and with probability $1-r$, Alice samples from its relabeling (red dice), but at the end of the protocol she ``forgets" where she sampled from. The resulting probability distribution $r\vect{p}+(1-r)\pi\vect{p}$ is more uncertain than the initial one associated with the blue (red) dice $\vect{p}$ ($\pi\vect{p}$). Color online.}
\label{fgr1}
\end{figure}

We formulate the above requirements quantitatively using Birkhoff's theorem \cite{Birkhoff:DoublyStochasticPermutations,Bhatia:MatrixAnalysis}, which states that the convex hull of permutation matrices is the class of doubly-stochastic matrices (their components are nonnegative real numbers, and each row and column sums to 1). Birkhoff theorem thus implies that a probability vector $\vect{q}$ obtained from $\vect{p}$ by a random relabeling is more uncertain than the latter if and only if the two are related by a doubly-stochastic matrix,
$\vect{q}=D\vect{p}$, which is equivalent to $\vect{q}\prec\vect{p}$. The last equation is known as a majorization relation \cite{MarshallOlkin:Majorization} and consists of a system of $d$ inequalities\footnote{ 
A vector $\vect{x}\in\RR^d$ \emph{is majorized by} a vector $\vect{y}\in\RR^d$, and write $\vect{x}\prec\vect{y}$, whenever 
$\sum_{j=1}^k x_j^{\downarrow}\leqslant\sum_{j=1}^k y_j^{\downarrow}$ for all $1\leqslant k\leqslant d-1$, with  $\sum_{j=1}^d x_j^{\downarrow}=\sum_{j=1}^d y_j^{\downarrow}$. The down-arrow notation denotes that the component of the corresponding vector are ordered in decreasing order,  $x_1^{\downarrow}\geqslant x_2^{\downarrow}\geqslant\cdots\geqslant x_d^{\downarrow}$.}. The above discussion  implies that any measure of uncertainty has to preserve the partial order induced by majorization. The class of functions that preserve this order are the Schur-concave functions. These are functions $\Phi$ on a $d$-dimensional probability space, $\Phi:\RR^d\longrightarrow \RR$, for which $\Phi(\vect{x})\geqslant\Phi(\vect{y})$ whenever $\vect{x}\prec\vect{y}$, $\forall \vect{x},\vect{y}\in\RR^d$. We therefore define a measure of uncertainty as being any non-negative Schur-concave function that takes the value zero on the vector $\vect{x}=(1,0,\ldots,0)$. The last requirement is not essential but is convenient as it ensures that the measure is non-negative.

Our definition for a measure of uncertainty is very general and resulted solely from requiring MURR; it also encompasses the most common entropy functions used in information theory, but it is not restricted to them. As we are not concerned with asymptotic regimes, we use in the following the most general $\Phi$ to quantify uncertainty, without making any assumptions about its functional form.

Having defined what a measure of uncertainty is, we now use it to study uncertainty relations. Let $\rho$ be a mixed state on a $d$-dimensional Hilbert space $\HC\cong\mathbb{C}^d$. For simplicity of the exposition, we first consider two basis (projective)  measurements.  We denote the two orthonormal bases of $\HC$ by $\{\ket{a_m}\}_{m=1}^{d}$ and $\{\ket{b_n}\}_{n=1}^{d}$. We also denote by $p_m(\rho)=\matl{a_m}{\rho}{a_m}$ and $q_n(\rho)=\matl{b_n}{\rho}{b_n}$ the two probability distributions obtained by measuring $\rho$ with respect to these bases. We collect the numbers $p_m(\rho)$ and $q_n(\rho)$ into two probability vectors $\vect{p}(\rho)$ and $\vect{q}(\rho)$, respectively. The goal of our work is to bound the uncertainty about $\vect{p}(\rho)$ and $\vect{q}(\rho)$ by a quantity that depends only on the bases elements but not on the state $\rho$. The object of our investigation is therefore the joint probability distribution $\vect{p}(\rho)\otimes\vect{q}(\rho)$.

\begin{figure}
\includegraphics[scale=0.4]{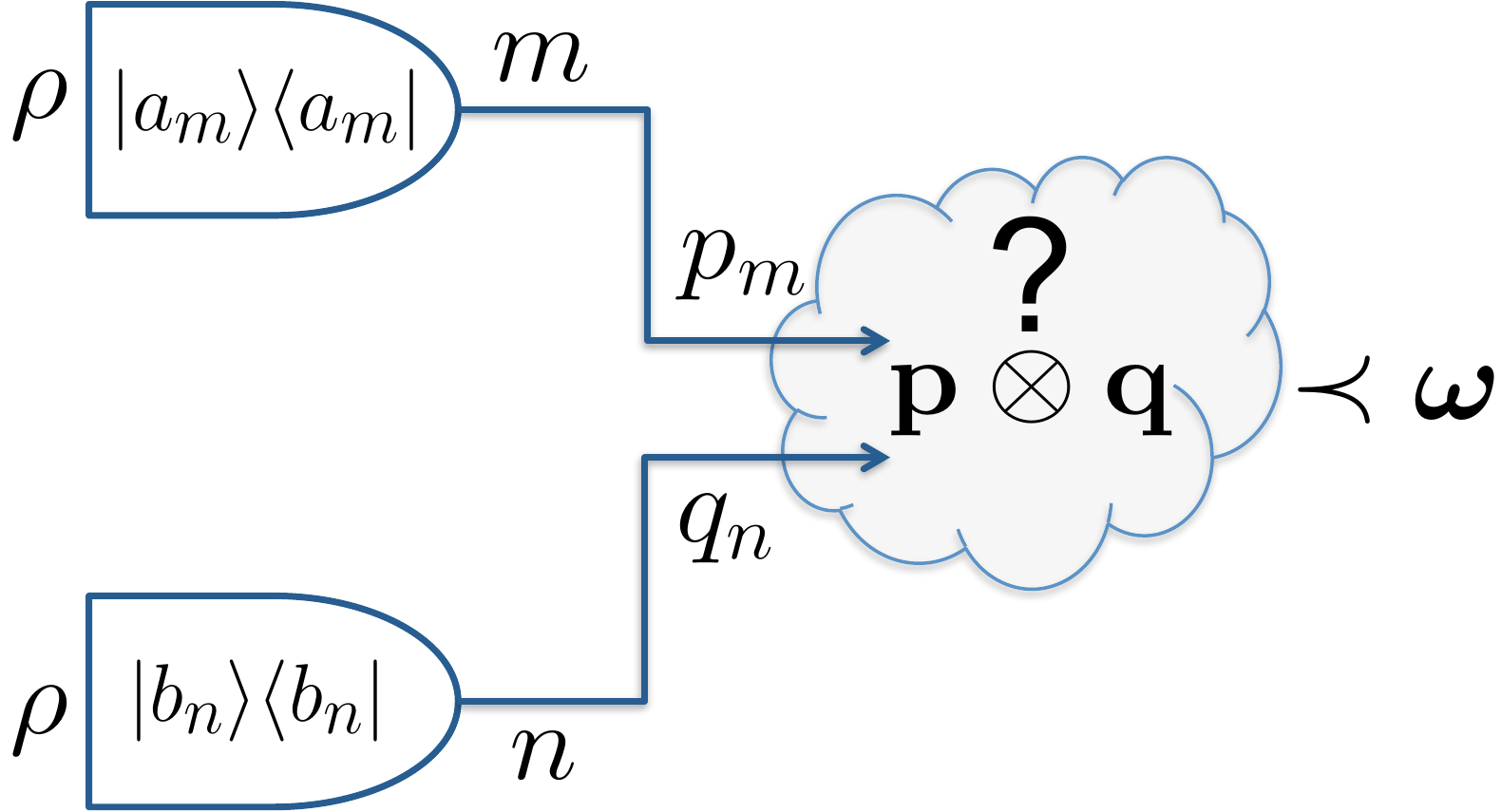}
\caption{A quantum state is measured using two orthonormal bases. We collect the induced joint probability distribution in a vector $\vect{p}\otimes\vect{q}$ and quantify its uncertainty in terms of a majorization relation, independently of the state $\rho$. Color online.}
\label{fgr2}
\end{figure}
The main result of our article is an uncertainty relation of the form
\begin{equation}\label{eqn4}
\vect{p}(\rho)\otimes\vect{q}(\rho)\prec\vect{\omega},\quad\forall \rho,
\end{equation}
where $\vect{\omega}$ is some vector independent of $\rho$ that we explicitly calculate. 
\textcolor{black}{A majorization uncertainty relation of a similar form was first introduced by Partovi in \cite{PhysRevA.84.052117}, however his right hand side of the majorization relation is not explicit but written in terms of supremum over all density matrices, which makes it difficult to calculate.}
We call \eqref{eqn4} a \emph{universal uncertainty relation} (UUR) as, for any measure of uncertainty $\Phi$,
\begin{equation}\label{eqn5}
\Phi\left(\vect{p}(\rho)\otimes\vect{q}(\rho) \right) \geqslant \Phi(\vect{\omega}),\quad\forall \rho.
\end{equation}
The UUR \eqref{eqn4} generates in fact an infinite family of uncertainty relations of the form \eqref{eqn5}, one for each $\Phi$. The right hand side of \eqref{eqn5} provides a single-number lower bound on the uncertainty of the joint measurement results. Whenever 
$\Phi$ is additive under tensor products (e.g. Renyi entropies, minus the logarithm of the G-concurrence \cite{PhysRevA.71.012318} or minus the logarithm of the minimum non-zero component of the probability distribution),  \eqref{eqn5} splits as
\begin{equation}\label{eqn6}
\Phi(\vect{p}(\rho))+\Phi(\vect{q}(\rho))\geqslant \Phi(\vect{\omega}).
\end{equation}

We now construct the $d^2$-dimensional vector $\vect{\omega}$ appearing on the right hand side of our UUR \eqref{eqn4}. 
Let $\IC_k\subset[d]\times[d]$ be a subset of $k$ distinct pair of indices $(m,n)$, where $[d]$ is the set of natural numbers ranging from $1$ to $d$. Let
\begin{equation}\label{eqn7}
\Omega_k:=\max_{\IC_k}\max_{\rho}\sum_{(m,n)\in\IC_k}p_m(\rho)q_n(\rho),
\end{equation}
where the outer maximum is over all subsets $\IC_k$ with cardinality $k$ and the inner maximum is taken over all density matrices. Then the vector $\vect{\omega}$ in the UUR \eqref{eqn4} is given by
\begin{equation}\label{eqn8}
\vect{\omega}=(\Omega_1,\Omega_2-\Omega_1,\ldots,\Omega_d-\Omega_{d-1},0,\ldots,0).
\end{equation}
Moreover, we show in the Appendix that  $\Omega_k=1$, for all $d\leqslant k\leqslant d^2$.

The quantities $\Omega_k$ in \eqref{eqn7} can be in general difficult to calculate explicitly, as they involve an optimization problem. However, in the Appendix we show that the first two elements can be computed explicitly as
\begin{align}\label{eqn9}
\Omega_1=\frac{1}{4}\left[ 1+c\right]^2\;,\quad
\Omega_2=\frac{1}{4}\left[ 1+c'\right]^2
\end{align}
where $c:=\max_{m,n}|\ip{a_m}{b_n}|$, and 
$$
c':=\max\sqrt{|\ip{a_m}{b_n}|^2+|\ip{a_{m'}}{b_{n'}}|^2}
$$ 
where the maximum is taken over all indexes $m=m'$ and $n\neq n'$, and over all indexes $n=n'$ and $m\neq m'$.

For $k>2$, we upper bound each $\Omega_k$ in \eqref{eqn7} by
\begin{align}
\widetilde\Omega_k&:=\max_{\substack{\RC,\SC\\|\RC|+|\SC|=k+1}}\max_{\rho}\left(\sum_{m\in\RC}p_m(\rho)\right) \left(\sum_{n\in\SC}q_n(\rho) \right)\label{eqn10}\\
&=\frac{1}{4}\max_{\substack{\RC,\SC\\|\RC|+|\SC|=k+1}}\left\Vert \sum_{m\in\RC}\dya{a_m} + \sum_{n\in\SC}\dya{b_n} \right\Vert^{2}_{\infty}\leqslant 1,\label{eqn11}
\end{align}
where $\RC$ ($\SC$) are subsets of distinct indices from $[d]$, $|\RC|$ ($|\SC|$) denotes the size (number of elements) of $\RC$ ($|\SC|$), and 
$\Vert\cdot\Vert_{\infty}$ denotes the infinity operator norm -- which, for positive operators (as it is in our case), coincides with the maximum eigenvalue of its argument. Moreover,
$\widetilde\Omega_d=1$. Note that $\Omega_k=\widetilde\Omega_k$ for $k=1,2$, and otherwise 
$\Omega_k\leq\widetilde\Omega_k$ since for a fixed $\rho$, all the terms in the sum of \eqref{eqn7} are strictly contained in the expression \eqref{eqn10}. The equality in~\eqref{eqn11} is non-trivial and follows from the main technical Theorem of this article (See Theorem~1 in the Appendix): 
$\max_{\rho}\Tr(\rho A)\Tr(\rho B)= \frac{1}{4}\Vert A+ B\Vert_\infty^2$, 
for two projections $A$ and $B$. 

Similar to the definition of the vector $\vect{\omega}$ in \eqref{eqn8},
we construct the vector $\vect{\widetilde\omega}$ as in~\eqref{eqn8} by replacing 
$\Omega_{k}$  with $\widetilde\Omega_k$. A simple calculation (see Appendix) shows that  
\begin{equation}\label{eqn12}
\vect{p}(\rho)\otimes\vect{q}(\rho)\prec\vect{\widetilde\omega},\quad\forall \rho.
\end{equation}
Therefore $\vect{\widetilde\omega}$ provides a (weaker) lower-bound for the UUR \eqref{eqn4}, but which is now explicitly computable. 

To appreciate the generality of our UUR \eqref{eqn4}, we compare in Fig.~\ref{fgr4} the best known lower bounds for the uncertainty of the measurement in two bases with our induced uncertainty relation \eqref{eqn5}, in which we take $\Phi$ to be the Shannon entropy $H$. We consider the region in which $c>0.83$, for which the best known bound \cite{PhysRevA.77.042110} has an explicit analytical form. We note that our bound over-performs \cite{PhysRevA.77.042110} in a large number of instances (around 90\% of the time).
For $c<0.83$, our bound tend to be slightly worse than \cite{PhysRevA.77.042110}, but this is expected since our uncertainty relation is valid for \emph{all} measures of uncertainty and is not optimized for a specific one such as Shannon's.
Next we take $\Phi=H_{\infty}$ in \eqref{eqn5} and note that we recover Maassen's and Uffink bound \cite{PhysRevLett.60.1103} for the minimum entropy, which is tight. Finally, choosing $\Phi=H_{\alpha}$ (Renyi-$\alpha$ entropy) in \eqref{eqn5} provides yet a novel entropic uncertainty relation valid for \emph{all} values of the parameter $\alpha$. 
 
\begin{figure}
\includegraphics[scale=0.55]{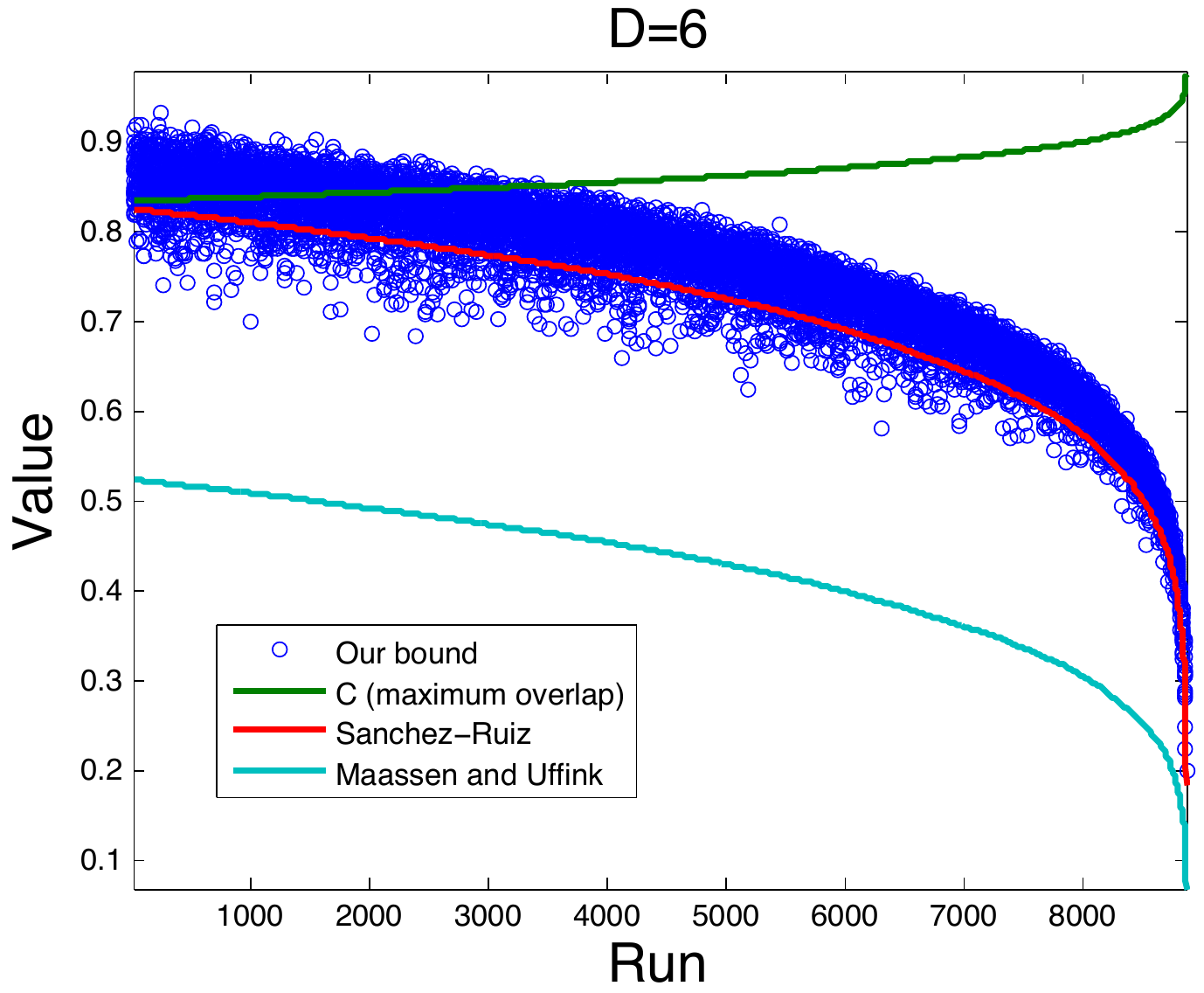}
\caption{For each run we randomly generate two orthonormal bases and a random state $\ket{\psi}$ in a 6-dimensional Hilbert space, then compute our lower bound $H(\vect\omega)$ (right hand side of \eqref{eqn5}, using  the Shannon entropy $H$ as a measure of uncertainty);  $c=\max_{m,n}|\ip{a_m}{b_n}|$ denotes the maximum overlap between the two bases. Color online.}
\label{fgr4}
\end{figure}

We now extend our results to the most general case of $L\geq2$ positive operator valued measures (POVMs). Denote by $\{\Pi_{\alpha_\ell}^{(\ell)}\}_{\alpha_\ell=1}^{N_\ell}$ the $\ell$-th POVM, with $1\leqslant \ell\leqslant L$.  The quantity $N_\ell$ denotes the number of elements in the $l$-th POVM, and the index $\alpha_\ell$ labels its elements, with $\alpha_\ell=1,2,...,N_\ell$. 
A measurement of $\rho$ with the $\ell$-th POVM $\Pi^{(\ell)}$ induces a probability distribution vector 
$\vect{p}^{(\ell)}(\rho)=({p}^{(\ell)}_{1}(\rho),{p}^{(\ell)}_{2}(\rho),...,{p}^{(\ell)}_{N_{\ell}}(\rho))$. We discover a UUR of the form
\begin{equation}\label{eqn13}
\bigotimes_{\ell=1}^L\vect{p}^{(\ell)}(\rho)\prec\vect{\omega},\quad\forall \rho,
\end{equation}
where the quantity on the left hand side represents the joint probability distribution induced by measuring $\rho$ with each POVM $\Pi^{(\ell)}$. 
Here 
\begin{equation}\label{eqn14}
\vect{\omega}=(\Omega_1,\Omega_2-\Omega_1,\Omega_3-\Omega_2,\ldots, \Omega_N-\Omega_{N-1}),
\end{equation}
where $N\equiv N_1N_2\cdots N_L$, and for  $k=1,2,...,N$
\begin{equation}\label{eqn15}
\Omega_k:=\max_{\mathcal{I}_k}\max_{\rho}\sum_{(\alpha_{1},...,\alpha_{L}) \in\mathcal{I}_k}{p}^{(1)}_{\alpha_{1}}(\rho){p}^{(2)}_{\alpha_{2}}(\rho)\cdots{p}^{(L)}_{\alpha_{L}}(\rho),
\end{equation}
where $\IC_k\subset[N_1]\times[N_2]\times\cdots\times[N_L]$ is a subset of $k$ distinct string of indices $(\alpha_{1},...,\alpha_{L})$ (here $[N_j]$ is the set of natural numbers ranging from $1$ to $N_j$).

Since the above quantities $\Omega_k$ can in general be difficult to calculate explicitly, we have found tight upper bounds 
$\widetilde\Omega_k$ that do not involve an optimization over all states $\rho$. Our upper bound 
$\Omega_k\leq\widetilde\Omega_k$ is given by
\begin{equation}\label{eqn16}
\widetilde\Omega_k:=\max_{\substack{\SC_1,\ldots,\SC_L  \\ \sum_{\ell=1}^{L}|\SC_\ell|=L+k-1 }}
\left\Vert \frac{1}{L} \sum_{\ell=1}^L \left( \sum_{\alpha_\ell\in\SC_\ell}\Pi^{(\ell)}_{\alpha_j} \right) \right\Vert_{\infty}^L \leqslant 1,
\end{equation} 
where $\SC_\ell$ denotes a subset of distinct indices from $[N_\ell]$, $|\SC_\ell|$ denotes the size (number of elements) of $\SC_\ell$. Note that by definition $\Omega_N=1$.
We define the vector $\vect{\widetilde\omega}$ as in~\eqref{eqn14} by replacing $\Omega_k$ with $\widetilde\Omega_k$,
and then show in the Appendix that the UUR 
in Eq.~\eqref{eqn13} holds with $\vect{\omega}$ replaced by $\vect{\widetilde\omega}$.

Note that an $L>2$ measurement uncertainty relation can be trivially generated by a summing pairwise two-measurement uncertainty relations, one for each pair of observables. Our UUR \eqref{eqn13} is more powerful and is not of this form. This fact can be seen most clearly in a set of measurement operators in which any two observables share a common eigenvector. In this case, a two-measurement uncertainty relation will provide a trivial lower bound of zero, hence the pairwise sum must also be zero. However, the vector $\vect{\omega}$ in \eqref{eqn14} is in general different from $(1,0,.\ldots,0)$, (see Example~1 in Sec.~B of the Supplementary Material), thus providing a non trivial bound for the UUR \eqref{eqn13} (or the induced family obtained by applying  various uncertainty measures $\Phi$ on it). Finally, the UUR \eqref{eqn13} is not restricted to MUBs or particular values of $L$, but is valid for any number of arbitrary bases.

To summarize, we derived two \textcolor{black}{explicit closed form} uncertainty relations, \eqref{eqn4} -- valid for measurements in two orthonormal bases, and  \eqref{eqn13} -- the generalization of the former to the most general setting of $L$ POVMs. Our relations are ``fine-grained"; they do not depend on a single number (such as the maximum overlap between bases elements), but on \emph{all} components of the vector $\vect{\widetilde\omega}$, which we compute explicitly, via a majorization relation.
Our uncertainty relations are \emph{universal} and  capture the essence of uncertainty in quantum mechanics, as they are not quantified by particular measures of uncertainty such as Shannon or Renyi entropies. 

We did not explore here which bases provide the most uncertain measurement results for the UURs. One may conjecture that MUBs are the suitable candidates. Indeed, this seems to be the case, and we conjecture that $\Omega_k$ in \eqref{eqn7} is given by $\Omega_k=\frac{1}{4}(1+\sqrt{k/d})^2$, which can then be used to construct (see \eqref{eqn8}) the vector $\vect{\omega}^{\mathrm{MUB}}$ for the UUR \eqref{eqn4}. The conjecture is strongly supported by numerical simulations. Moreover, we observed that for bases that are not MUBs, the best $\vect{\omega}$ we were able to find (numerically) always majorizes $\vect{\omega}^{\mathrm{MUB}}$, i.e. $\vect{\omega}^{\mathrm{MUB}}\prec\vect{\omega}$. This provides strong support for the initial assumption that MUBs provide the most uncertain measurement outcomes. 

Another important direction of investigation is the extension of the results presented here to uncertainty relations with quantum memory \cite{Berta:2010}. These are particularly useful in the context of quantum cryptography. However, such an extension is non-trivial and is left for future work.

The authors thank Robert Spekkens and Marco Piani for useful comments and discussions.
V. Gheorghiu and G. Gour acknowledge support from the Natural Sciences
and Engineering Research Council (NSERC) of Canada
and from the Pacific Institute for the Mathematical Sciences
(PIMS). Shmuel Friedland was partially supported by NSF grant
 DMS-1216393. 
 
\textit{Note added:} Soon after the first version of this manuscript appeared on the arXiv, Pucha\l{}a, Rudnicki and \.Zyczkowski submitted a similar result on the arXiv \cite{quantph:1304.7755} (valid for orthonormal bases), now published in \cite{ZycUUR2013}. Their proof uses different techniques and we refer the interested reader to their paper for more details.


\begin{titlepage}
\center{\large\textbf{Supplementary Information\\Universal Uncertainty Relations}\\}
\center{~{ }\\}

\end{titlepage}

\onecolumngrid

\appendix

Here we derive and elaborate on the main results discussed in the paper. 
The Appendix consists of three parts. In Sec.~\ref{Asct1} we 
prove the results for two orthonormal bases. In Sec.~\ref{Asct2} we prove the results for $L>2$ POVMs. In Sec.~\ref{Asct3} we discuss a possible generalization of our results to the case where each measurement is applied according to a pre-probability distribution; we call such relations \emph{weighted uncertainty relations}.

\section{Two orthonormal bases}\label{Asct1}

We first prove (4) (main text). By construction, the sum of the first $k$ largest components of $\vect{p}(\rho)\otimes\vect{q}(\rho)$ in~(4) (main text) can not be greater than $\Omega_k$. Observe also that $\Omega_k=\Omega_1+(\Omega_2-\Omega_1)+\ldots+(\Omega_k-\Omega_{k-1})$ is the sum of the first $k$ components of the \emph{un-ordered} vector $(\Omega_1,\Omega_2-\Omega_1,\ldots,\Omega_d-\Omega_{d-1},0,\ldots,0)$, and this sum cannot be greater than the sum of the first $k$ components of the \emph{ordered} vector $(\Omega_1,\Omega_2-\Omega_1,\ldots,\Omega_d-\Omega_{d-1},0,\ldots,0)^\downarrow$, since the components of the latter are ordered in decreasing order; hence the UUR~(4) (main text) is proven.

The fact that $\Omega_k=1$, for all $d\leqslant k\leqslant d^2$, follows from the observation that, for orthonormal bases, one can always find a state $\rho$ for which $\vect{p}=(1,0,\ldots,0)$ (choose the state to be the projector onto one of the basis elements, e.g. $\rho=\dya{a_1}$). In this case, the expression in (7) of the main text for $\Omega_k$ achieves its maximum possible value of 1, for all $d\leqslant k\leqslant d^2$.

We next prove the main Theorem, which we then use to compute the first two components $\Omega_1$ and $\Omega_2$ in (9) of the main text.

\begin{theorem}\label{Athm1}
Let $A$ and $B$ be two orthogonal projectors, $A=A^2=A^\dagger$, $B=B^2=B^\dagger$, acting on $\mathbb{C}^d$. Then, 
\begin{equation}\label{Aeqn1}
\max_\rho\Tr(\rho A)\Tr(\rho B)=\frac{1}{4}\Vert A+B\Vert^2_{\infty}\leqslant 1,
\end{equation}
where the maximum is taken over all density matrices (i.e. positive-semidefinite matrices with trace 1) acting on $\mathbb{C}^d$,
and $\Vert\cdot\Vert_\infty$ denotes the operator norm (the maximum singular value of its argument).
\end{theorem}
In order to prove Theorem~\ref{Athm1} we first prove a weaker version of it in the following lemma.

\begin{lemma}\label{Alma1}
Let $\dya{a}$ and $\dya{b}$ be two rank-one projectors acting on $\mathbb{C}^d$. Then,
\begin{equation}\label{Aeqn2}
\max_\rho \matl{a}{\rho}{a}\matl{b}{\rho}{b}=\frac{1}{4}(1+\vert \ip{a}{b}\vert)^2=\frac{1}{4}\left\Vert \dya{a}+\dya{b}\right\Vert_\infty^2,
\end{equation}
where the maximum is taken over all $d\times d$ density matrices $\rho$.  
\end{lemma}
\begin{proof}
First note that
\begin{equation}\label{Aeqn3}
\matl{a}{\rho}{a}\matl{b}{\rho}{b}=\left(\sqrt{\matl{a}{\rho}{a}\matl{b}{\rho}{b}}\right)^2\leqslant\frac{1}{4} \left(\Tr\left[\rho (\dya{a}+\dya{b}) \right]\right)^2\leqslant\frac{1}{4}\left\Vert\dya{a}+\dya{b}\right\Vert_{\infty}^{2},
\end{equation}
where the first inequality follows from the geometric-arithmetic mean inequality and the second inequality from the standard property of the operator norm.

We next construct a state $\rho$ for which the inequality above is saturated. 
We show that the maximum can be achieved with a rank one matrix, $\rho=|x\rangle\langle x|$. 
Let $\ket{\tilde a}$ be a normalized vector such that the set $\{\ket{\tilde a},\ket{b}\}$ is orthonormal and the vector $\ket{a}$ is in the span of $\{\ket{\tilde a},\ket{b}\}$ (this can be done using the Gram-Schmidt process). We can therefore assume that also the optimal $\ket{x}$ is in the span of  $\{\ket{\tilde a},\ket{b}\}$ and expand it as
\begin{equation}\label{Aeqn4}
\ket{x}=\ee^{\ii\phi}r\ket{b}+\sqrt{1-r^2}\ee^{\ii\theta}\ket{\tilde a},
\end{equation}
with $0\leqslant r\leqslant 1$ and $\theta,\phi$ real numbers. 
We choose $\phi$ such that $\ee^{\ii\phi}\ip{a}{b}=\vert\ip{a}{b}\vert\equiv c$ and take $\theta$ so that $\ee^{\ii\theta}\ip{a}{\tilde a}=\vert \ip{a}{\tilde a}\vert=\sqrt{1-c^2}$.
We then obtain
\begin{equation}\label{Aeqn5}
\matl{a}{\rho}{a}\matl{b}{\rho}{b}=\vert\ip{a}{x}\vert^2\vert\ip{b}{x}\vert^2=r^2(rc+\sqrt{1-r^2}\sqrt{1-c^2})^2.
\end{equation}
To optimize the function above over the variable $r$ we denote $\cos(\alpha)\equiv r$ and $\cos(\beta)\equiv c$. Note that
without loss of generality we can take $0\leq\alpha,\beta\leq\pi/2$. Therefore,
\begin{equation}\label{Aeqn6}
\matl{a}{\rho}{a}\matl{b}{\rho}{b}=\cos^2(\alpha)\left[\cos(\alpha)\cos(\beta)+\sin(\alpha)\sin(\beta)\right]^2
=\cos^2(\alpha)\cos^2(\alpha-\beta)=\frac{1}{4}\left[\cos(2\alpha-\beta)+\cos(\beta)\right]^2.
\end{equation}
Since $\alpha$ is determined by $r$, we can choose $r$ such that $\alpha=\beta/2$. For this choice of $r$ we get
\begin{equation}\label{Aeqn7}
\matl{a}{\rho}{a}\matl{b}{\rho}{b}=\frac{1}{4}(1+c)^2,
\end{equation}
which proves the first equality in \eqref{Aeqn2}. The 
second equality in \eqref{Aeqn2} follows from direct calculation.
\end{proof}

\textbf{Proof of Theorem~\ref{Athm1}}:
Note first that 
\begin{equation}\label{Aeqn8}
\max_\rho\Tr(\rho A)\Tr(\rho B)\leqslant \max_\rho\frac{1}{4}\left(\Tr\left[\rho(A+B)\right]\right)^2\leqslant\frac{1}{4} \Vert A+B \Vert_\infty^2\leqslant 1,
\end{equation}
where the first inequality follows from the geometric-arithmetic mean inequality and the second one from $\tr [\rho(A+B)]\leqslant
 \sum_{i=1}^d \lambda_i(\rho)\lambda_i(A+B)\leqslant \|A+B\|_{\infty}
 \tr(\rho)=\|A+B\|_{\infty}$, where $\lambda_i(\bullet)$ denotes the $i$-th eigenvalue of its argument. The last inequality in \eqref{Aeqn8} is trivial and follows at once from the triangle inequality for the norm, $\Vert A+B\Vert_\infty\leqslant \Vert A \Vert_\infty + \Vert B \Vert_\infty=2$ (since $A$ and $B$ are projectors). 

It is therefore left to show that the maximum in~\eqref{Aeqn1} can be achieved. In fact, we will show now that it can be achieved 
with a pure state $\rho=|x\rangle\langle x|$; that is,
\begin{equation}\label{Aeqn9}
\max_{\ket{x}}\matl{x}{A}{x}\matl{x}{B}{x}=\frac{1}{4}\Vert A+B \Vert^2_\infty\leqslant 1.
\end{equation}

To see it, let $A\HC$ and $B\HC$ be the subspaces onto which $A$ and $B$ project, respectively. 
Now, if $A\HC\bigcap B\HC\neq\{\vect{0}\}$ (i.e. the subspaces have a non-trivial intersection) then there exist some $\ket{x_0}$ that belongs to both $A\HC$ and $B\HC$, and therefore $A\ket{x_0}=\ket{x_0}$ and $B\ket{x_0}=\ket{x_0}$. Hence, in this case,
$\Vert{A+B}\Vert_{\infty}=2$ and the maximum in~\eqref{Aeqn9} is achieved for $\ket{x}=\ket{x_0}$. 
We therefore assume in the rest of the proof that $A\HC\bigcap B\HC=\{\vect{0}\}$.

Since $A$ and $B$ are projectors, 
$\Vert A+B\Vert_{\infty}\geqslant 1$.
We first prove the theorem for the simpler case in which $\Vert A+B\Vert_{\infty}= 1$.
In this case $AB=BA=0$. To see it note that $\forall \ket{y}\in A\HC$, $\matl{y}{A+B}{y}=1+\matl{y}{B}{y}=1$ (since the norm is 1) and therefore $\matl{y}{B}{y}=0$ $\forall \ket{y}\in A\HC$. This implies that $BA\ket{v}=0$ $\forall\ket{v}\in\HC$ and therefore $BA=0$. 
Similarly, $AB=0$. 
We need to show now that whenever $\Vert A+B\Vert_{\infty}= 1$ the maximum value of the left hand side of \eqref{Aeqn9} is 1/4. One can check this by writing $\ket{x}=\alpha\ket{x_A}+\beta\ket{x_B}$  where $\ket{x_A}\in A \HC$ and  $\ket{x_B}\in B \HC$, with $\ip{x_A}{x_B}=0$. The product on the left hand side of \eqref{Aeqn9} becomes $|\alpha|^2|\beta|^2$, with maximum value 1/4. The right hand side is trivially equal to $1/4$, and the Theorem is therefore proved for this case.

The last case that is left to prove is when $A\HC\bigcap B\HC=\{\vect{0}\}$ and $\Vert A+B\Vert_{\infty}>1$. We define $\ket{y}$ to be the vector for which 
\begin{equation}\label{Aeqn10}
(A+B)\ket{y}=(\Vert A+B \Vert_{\infty})\ket{y}, 
\end{equation}
and let 
\begin{equation}\label{Aeqn11}
\ket{a}=\frac{A\ket{y}}{\Vert A\ket{y}\Vert}, \quad \ket{b}=\frac{B\ket{y}}{\Vert B\ket{y} \Vert}.
\end{equation}
Note that both $\Vert A\ket{y} \Vert$ and $\Vert B\ket{y} \Vert$ are non-zero since $\Vert A+B\Vert_\infty>1$, so $A\ket{y}\neq0$ and similarly $B\ket{y}\neq0$.
We next prove that
\begin{equation}\label{Aeqn12}
\Vert A+B \Vert_\infty = \left\Vert \dya{a}+\dya{b} \right\Vert_\infty.
\end{equation}
Indeed, since $\ket{a}\ip{a}{y}=A\ket{y}$ and
$\ket{b}\ip{b}{y}=B\ket{y}$ we get 
\begin{equation}\label{Aeqn13}
(\dya{a}+\dya{b})\ket{y}=(A+B)\ket{y}=(\Vert A+B \Vert_\infty) \ket{y},
\end{equation}
which implies 
\begin{equation}\label{Aeqn14}
\Vert \dya{a}+\dya{b} \Vert_{\infty} \geqslant \Vert A+B \Vert_\infty.
\end{equation}
On the other hand,
\begin{equation}\label{Aeqn15}
\Vert \dya{a}+\dya{b} \Vert_{\infty}\leqslant \Vert A+B \Vert_\infty,
\end{equation}
since $A\geqslant\dya{a}$ and $B\geqslant\dya{b}$ (see the definition of $\ket{a}$ and $\ket{b}$ in \eqref{Aeqn11}).
Combining \eqref{Aeqn14} and \eqref{Aeqn15} yields the equality in~\eqref{Aeqn12}.

Since $A\geqslant \dya{a}$ and $B\geqslant \dya{b}$, we also have $\Tr(\rho A)\Tr(\rho B)\geqslant \matl{a}{\rho}{a}\matl{b}{\rho}{b}$, and taking the maximum over all $\rho$ yields
\begin{equation}\label{Aeqn16}
\max_\rho\Tr(\rho A)\Tr(\rho B)\geqslant \max_\rho\matl{a}{\rho}{a}\matl{b}{\rho}{b}=\frac{1}{4}\Vert \dya{a}+\dya{b} \Vert_{\infty}^2 = \frac{1}{4}\Vert A+B \Vert_\infty^2,
\end{equation}
where the first equality follows from Lemma~\ref{Alma1} and the second equality from \eqref{Aeqn12}.
Combining \eqref{Aeqn8} and \eqref{Aeqn16} completes the proof of Theorem~\ref{Athm1}. $\Box$

\subsection{Exact expressions for $\Omega_1$ and $\Omega_2$}

We now have all ingredients for computing $\Omega_1$ and $\Omega_2$. Note that
\begin{align}\label{Aeqn17}
\Omega_1&=\max_{m,n}\max_\rho p_m(\rho)q_n(\rho)=\max_{m,n}\max_{\rho}\Tr(\rho\dya{a_m})\Tr(\rho\dya{b_n})\notag\\
&=\max_{m,n}\frac{1}{4}\left\Vert\dya{a_m}+\dya{b_n} \right\Vert^2_{\infty}=\max_{m,n}\frac{1}{4}\left[1+\vert \ip{a_m}{b_n}\vert\right]^2
=\frac{1}{4}[1+c]^2,
\end{align}
where the first equality on the second line follows from Theorem~\ref{Athm1} and the second equality
 on the second line follows from the arguments in the last part of
 the proof of Theorem~\ref{Athm1}.

To compute $\Omega_2$, first observe that, for fixed $m$ and $k$,
\begin{align}\label{Aeqn18}
\max_\rho&\left[ p_m(\rho)\sum_{n=1}^kq_n(\rho)\right]=\max_\rho\left[ \Tr\left(\rho\dya{a_m}\right)\Tr\left(\rho\sum_{n=1}^k\dya{b_n}\right) \right]\notag\\
&=\frac{1}{4}\left\Vert \dya{a_m}+\sum_{n=1}^k\dya{b_n}\right\Vert^2_{\infty}=\frac{1}{4}\left( 1+\sqrt{\sum_{n=1}^k\vert \ip{a_m}{b_n} \vert^2} \right)^2,
\end{align}
where the third equality follows from Theorem~\ref{Athm1} and the last one from the arguments in the last part of the proof of Theorem~\ref{Athm1}. By symmetry, the above equations hold under interchanging  $p$ and $q$.
Next note that
\begin{equation}\label{Aeqn19}
p_mq_n+p_{m'}q_{n'}\leqslant(q_n+q_{n'})\max\{p_m,p_{m'}\}
\end{equation}
and also
\begin{equation}\label{Aeqn20}
p_mq_n+p_{m'}q_{n'}\leqslant(p_m+p_{m'})\max\{q_n,q_{n'}\},
\end{equation}
where, for simplicity of notation, we dropped the $\rho$ dependence. Combining \eqref{Aeqn18}, \eqref{Aeqn19} and \eqref{Aeqn20} yields the expression (9) (main text) for $\Omega_2$.

\subsection{The Upper Bounds $\widetilde\Omega_k$}

Theorem~\ref{Athm1} can also be used to prove the exact expressions introduced in (11)
for the upper bounds defined in (10) (main text). In particular, (11) follows from (10)
by substituting 
$$
A=\sum_{m\in\RC}\dya{a_m} \;\;\;\text{and}\;\;\;B=\sum_{n\in\SC}\dya{b_n}
$$
in Theorem~\ref{Athm1}.

The upper bounds $\widetilde\Omega_k$ in (10) can be easily visualized using a Young-like diagram. 
\begin{figure}
\includegraphics[scale=0.6]{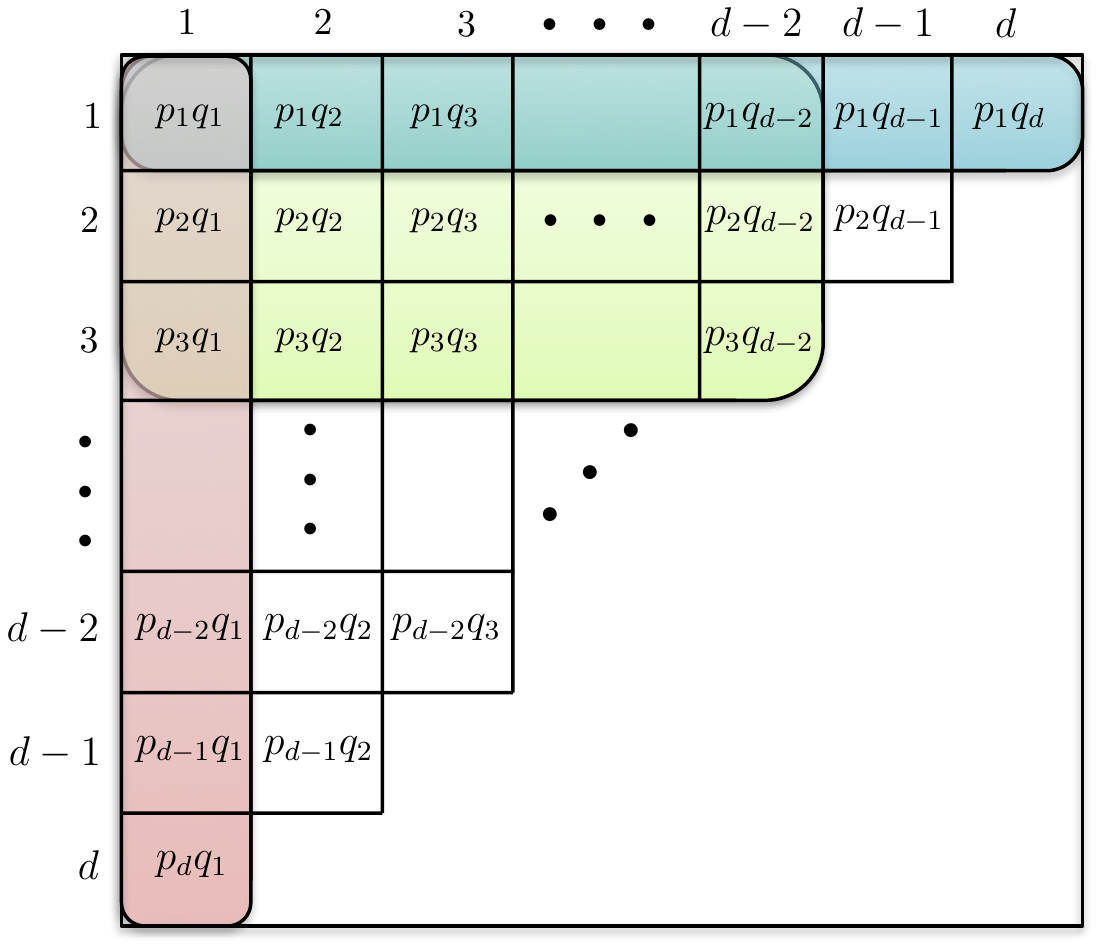}
\caption{Upper bound $\widetilde\Omega_k$ for $\Omega_k$ from a Young-like diagram. Color online.
}
\label{Afgr1}
\end{figure}
Each box in the figure represents a component $p_mq_n$ of the probability vector $\vect{p}(\rho)\otimes\vect{q}(\rho)$, arranged in lexicographical order (we assume for simplicity that $p_1\geqslant p_2\geqslant\cdots\geqslant p_D$ and  $q_1\geqslant q_2\geqslant\cdots\geqslant q_D$). For a fixed $\rho$, the upper bound $\widetilde\Omega_k$ in (10) (main text) corresponds to taking the maximum among the sum of the components of each $1\times D$ , $2\times (D-1)$ , $\ldots$ , $D\times 1$ rectangles in the figure, for a total of $D$ such rectangles. We highlight 3 possible rectangles using different colours : $1\times D$ with light blue, $3\times (D-2)$ with light green and $D\times 1$ with light red. It is easy to see that the vector $\Omega_k$ in (7) (main text) is upper bounded by $\widetilde\Omega_k$, as the former is always constructed as a partial sum over some rectangle in the figure. 

To prove the UUR~(12) (main text), we use essentially the same argument as in proving~(4) (main text), but we repeat it for the sake of completeness. By construction, $\Omega_k\leqslant\widetilde\Omega_k$, and the sum of the first $k$ largest components of $\vect{p}(\rho)\otimes\vect{q}(\rho)$ in~(4) (main text) can not be greater than $\Omega_k$. Now observe that $\widetilde\Omega_k=\widetilde\Omega_1+(\widetilde\Omega_2-\widetilde\Omega_1)+\ldots+(\widetilde\Omega_k-\widetilde\Omega_{k-1})$ is the sum of the first $k$ components of the \emph{un-ordered} vector $(\widetilde\Omega_1,\widetilde\Omega_2-\widetilde\Omega_1,\ldots,\widetilde\Omega_d-\widetilde\Omega_{d-1},0,\ldots,0)$, and this sum cannot be greater than the sum of the first $k$ components of the \emph{ordered} vector $(\widetilde\Omega_1,\widetilde\Omega_2-\widetilde\Omega_1,\ldots,\widetilde\Omega_d-\widetilde\Omega_{d-1},0,\ldots,0)^\downarrow$, since the components of the latter are ordered in decreasing order; hence the UUR~(12) is proven.

\section{$L\geqslant 2$ POVMs}\label{Asct2}
In the last part of the paper we extended the results valid for two orthonormal bases to the most general case of $L\geqslant 2$ POVMs, and derived the uncertainty relation (13) (main text). As was the case before for two orthonormal bases, we do not have an explicit expression for $\vect{\omega}$, however we use the upper bound (16) for $\Omega_k$ to construct another vector $\vect{\widetilde\omega}$ that also satisfies the uncertainty relation (13) (main text). We now prove that indeed (16)  is an upper bound on $\Omega_k$. Note that
\begin{align}\label{Beqn1}
&\max_{\mathcal{I}_k}\max_{\rho}\sum_{(\alpha_{1},...,\alpha_{L}) \in\mathcal{I}_k}{p}^{(1)}_{\alpha_{1}}(\rho){p}^{(2)}_{\alpha_{2}}(\rho)\cdots{p}^{(L)}_{\alpha_{L}}(\rho)=\max_{\mathcal{I}_k}\max_{\rho}\sum_{(\alpha_{1},...,\alpha_{L}) \in\mathcal{I}_k}\prod_{\ell=1}^L \Tr\left[\rho\Pi_{\alpha_{\ell}}^{(\ell)}\right]\notag\\
&\leq\max_{\substack{\SC_1,\ldots,\SC_L  \\ \sum_{\ell=1}^{L}|\SC_\ell|=L+k-1}}\max_{\rho}
\prod_{\ell=1}^L \Tr\left[\rho\sum_{\alpha_\ell\in\SC_\ell}\Pi_{\alpha_{\ell}}^{(\ell)}\right]
\leqslant\max_{\substack{\SC_1,\ldots,\SC_L  \\ \sum_{\ell=1}^{L}|\SC_\ell|=L+k-1 }}\max_{\rho}\left(\Tr\left[ \rho \left(\frac{1}{L}\sum_{\ell=1}^L \sum_{\alpha_\ell\in\SC_\ell} \Pi_{\alpha_\ell}^{(\ell)} \right)\right]\right)^L\notag\\
&=\max_{\substack{\SC_1,\ldots,\SC_L  \\ \sum_{\ell=1}^{L}|\SC_\ell|=L+k-1 }}\left\Vert \frac{1}{L} \sum_{\ell=1}^L \left( \sum_{\alpha_\ell\in\SC_\ell}\Pi^{(\ell)}_{\alpha_j} \right) \right\Vert_{\infty}^L \leqslant 1.
\end{align}
The first inequality is a simple consequence of the fact that $\mathcal{I}_{k}\subset \SC_1\times\cdots\times\SC_L$.
The second inequality follows from the geometric-arithmetic mean inequality, the third inequality follows from the properties of the infinity norm, and the last inequality follows from the fact that the operator inside the norm is smaller or equal to the identity. Therefore $\Omega_k\leqslant\widetilde\Omega_k$. Unlike the case of two orthonormal bases, where the first two upper bounds satisfy $\Omega_1=\widetilde\Omega_1$ and $\Omega_2=\widetilde\Omega_2$, for the case of POVMs this is no longer true, and in general the upper bounds $\widetilde\Omega_k$ are strict for \emph{all} $k$. The reason behind this fact is that Theorem~\ref{Athm1} does not admit a generalization to more than two measurements.

In the following example we show that our $L>2$ POVM uncertainty relation (13) is much stronger than the one obtained  by  summing pairwise two-measurements uncertainty relations, one for each pair of observables. 
\begin{example}\label{Bexmp1}
Consider the following three orthonormal bases in a four dimensional Hilbert space
\begin{align}\label{Beqn2}
\BC_1&=\left\{\boxed{\ket{0}}, \boxed{\ket{1}}, \ket{2}, \ket{3}\right\}\notag\\
\BC_2&=\left\{\boxed{\ket{0}}, \boxed{\frac{\ket{2}+\ket{3}}{\sqrt{2}}}, \frac{\ket{1}+\ket{2}-\ket{3}}{\sqrt{3}}, \frac{2\ket{1}-\ket{2}+\ket{3}}{\sqrt{6}}\right\}\notag\\
\BC_3&=\left\{\boxed{\frac{\ket{2}+\ket{3}}{\sqrt{2}}}, \boxed{\ket{1}}, \frac{\ket{0}+\ket{2}-\ket{3}}{\sqrt{3}}, \frac{2\ket{0}-\ket{2}+\ket{3}}{\sqrt{6}}\right\}.
\end{align}
Note that each pair of bases share a common eigenvector (boxed above), and the three common eigenvectors are all different: $\ket{0}, \frac{\ket{2}+\ket{3}}{\sqrt{2}}$ and $\ket{1}$. For this case, the upper bounds in (16) are $\widetilde\Omega_1 \cong 0.78$ and $\widetilde\Omega_k=1$, for $k>1$. Therefore $\vect{\widetilde\omega}=(0.78,0.22,0,\ldots,0)$, and our uncertainty relation (13) is non-trivial, whereas a sum of pairwise two-measurements uncertainty relation provides a trivial lower bound of zero (as each pair of basis share a common eigenvector).
\end{example}

\section{Weighted uncertainty relations}\label{Asct3}

Entropic uncertainty relations of the form
\begin{equation}\label{Ceqn1}
\sum_{\ell=1}^{L}t_\ell H(\vect{p}^{(\ell)}(\rho))\geq C.
\end{equation}
were also considered in the literature~(see \cite{1367-2630-12-2-025009} and the references within). Here $H$ is the Shannon or Renyi entropy,
$\sum_{\ell=1}^{L}t_\ell=1$, $t_\ell\geq 0$ is a probability associated with the measurement $\ell$, and $C$ is some positive constant. This is a generalization of the original entropic uncertainty relation~(3) (main text) to the case in which there are $L$ measurements not all with the same weight. In this section we incorporate this unevenness in the weights of the distinct measurements into our UURs formalism, and show that the UUR given in (13) of the main text also produces
entropic uncertainty relations of the form~\eqref{Ceqn1}.

Indeed, consider $L$ integers $\{w_\ell\}_{\ell=1,...,L}$ and set  $W\equiv\sum_{\ell=1}^Lw_\ell$.
Consider also $W$ measurements (not all distinct) such that for each $\ell$, $w_\ell$ of them are identical.
Thus, in this case of such $W$ measurements, (13) (main text) takes the form:
\begin{equation}\label{Ceqn2}
\bigotimes_{\ell=1}^L\left[\vect{p}^{(\ell)}(\rho)\right]^{\otimes w_\ell}\prec\vect{\omega},\quad\forall \rho.
\end{equation}
Now, by 
applying an additive measure of uncertainty $\Phi$ on \eqref{Ceqn2} we get
\begin{equation}\label{Ceqn3}
\sum_{\ell=1}^L w_l\Phi\left[\vect{p}^{(\ell)}(\rho)\right]\geqslant \Phi(\vect{\omega}),
\end{equation}
or, equivalently,
\begin{equation}\label{Ceqn4}
\sum_{\ell=1}^L t_\ell\Phi\left[\vect{p}^{(\ell)}(\rho)\right]\geqslant \frac{1}{W}\Phi(\vect{\omega})\;.
\end{equation} 
The quantities $t_\ell\equiv w_\ell/W$ form a probability distribution and the uncertainty relation \eqref{Ceqn4} can thus be seen as a ``weighted" one. Thus, the vector uncertainty relation~\eqref{Ceqn2} is a universal version of the entropic uncertainty relation~\eqref{Ceqn1}. Note, however, that the UUR given in~\eqref{Ceqn2} is not a new one but simply follows from our main result in (13) (main text) when applied to the special case where not all of the measurements are distinct.

\end{document}